\documentclass{sig-alternate}
\usepackage{soul}
\usepackage{amsmath}
\usepackage{amssymb}
\usepackage{subfigure}
\usepackage{url}
\usepackage{color}
\usepackage{enumitem}
\definecolor{darkgreen}{rgb}{0,0.5,0}
\usepackage[colorlinks=true,%
  linkcolor=red,%
  citecolor=darkgreen,%
  urlcolor=blue]{hyperref}
\usepackage{rotating}
\usepackage[noend,ruled,vlined]{algorithm2e}

\usepackage[table]{xcolor}
\definecolor{light-gray}{gray}{0.9}

\newcommand{\Am}{{\cal A}}

\newcommand{\dis}{{d}}
\newcommand{\rep}{{rep}}
\newcommand{\cc}{{\tt cc}\xspace}
\newcommand{\far}{{\tt far}\xspace}
\newcommand{\card}{{\tt R}}
\newcommand{\sumcard}{{\tt RF}}
\newcommand{\union}{\cup}

\newcommand{\ccalg}{\textsc{CC}\xspace}
\newcommand{\ccalgb}{\textsc{CC-B}\xspace}
\newcommand{\ccalgbl}{\textsc{CC-BL}\xspace}
\newcommand{\ccalgbli}{\textsc{CC-BLI}\xspace}
\newcommand{\ccalgblih}{\textsc{CC-BLIH}\xspace}
\newcommand{\cv}{{\mathbf c}}

\newtheorem{theorem}{Theorem}[section]
\newtheorem{lemma}[theorem]{Lemma}

\newtheorem{corollary}[theorem]{Corollary}
\newtheorem{definition}[theorem]{Definition}

\begin{document}


\title{Incremental Algorithms for Network Management and Analysis based on Closeness Centrality}


\author{Ahmet Erdem Sar{\i}y\"{u}ce$^{1,2}$, Kamer
  Kaya$^1$, Erik Saule$^1$, \"{U}mit V. \c{C}ataly\"{u}rek$^{1,3}$\\\\
{\normalsize Depts. $^1$Biomedical Informatics, $^2$Computer
  Science and Engineering, $^3$Electrical and Computer Engineering}\\
  {\normalsize The Ohio
  State University}\\ \small Email:\textit{sariyuce.1@osu.edu, \{kamer,esaule,umit\}@bmi.osu.edu}}

\maketitle 

\begin{abstract} Analyzing networks requires complex algorithms
to extract meaningful information. Centrality metrics have shown to be
correlated with the importance and loads of the nodes in network
traffic. Here, we are interested in the problem of centrality-based
network management. The problem has many applications such as
verifying the robustness of the networks and controlling or improving
the entity dissemination. It can be defined as finding a small set of
topological network modifications which yield a desired closeness
centrality configuration. As a fundamental building block to tackle
that problem, we propose incremental algorithms which efficiently
update the closeness centrality values upon changes in network
topology, i.e., edge insertions and deletions. Our algorithms are
proven to be efficient on many real-life networks, especially on
small-world networks, which have a small diameter and a spike-shaped
shortest distance distribution. In addition to closeness centrality,
they can also be a great arsenal for the shortest-path-based
management and analysis of the networks.  We experimentally validate
the efficiency of our algorithms on large networks and show that they
update the closeness centrality values of the temporal DBLP-coauthorship
network of 1.2 million users 460 times faster than it would take to
compute them from scratch. To the best of our knowledge, this is
the first work which can yield practical large-scale network management
based on closeness centrality values.
\end{abstract}




\section{Introduction}

Centrality metrics, such as closeness or betweenness, quantify how central a
node is in a network. They have been successfully used to carry analysis for
various purposes such as structural analysis of knowledge
networks~\cite{Pham10,Shi10}, power grid contingency analysis~\cite{Jin10},
quantifying importance in social networks~\cite{LeMerrer2009}, analysis of
covert networks~\cite{Krebs02}, decision/action networks~\cite{simsekb08},
and even for finding the best store locations in cities~\cite{Porta09}. Several
works which have been conducted to rapidly compute these metrics exist in the
literature. The algorithm with the best asymptotic complexity to compute
centrality metrics~\cite{brandes2001} is believed to be asymptotically
optimal~\cite{Kintali08}. Research have focused on either approximation
algorithms for computing centrality metrics~\cite{Chan09,Eppstein01,Okamoto08}
or on high performance computing techniques~\cite{Madduri2009,Shi11}. Today, it
is common to find large networks, and we are always in a quest for better
techniques which help us while performing centrality-based
analysis on them.



When the network topology is modified, ensuring the correctness of the
centralities is a challenging task. This problem has been studied for
dynamic and streaming networks~\cite{Green2012,Lee2012}. Even for some
applications involving a static network such as the contingency analysis
of power grids and robustness evaluation of networks, to be prepared
and take proactive measures, we need to know how the centrality values
change when the network topology is modified by an adversary and outer
effects such as natural disasters.

A similar problem arises in network management for which not only
knowing but also setting the centrality values in a controlled manner
via topology modifications is of concern to speed-up or contain the
entity dissemination. The problem is hard: there are $m$ candidate
edges to delete and $\mathcal{O}(n^2)$ candidate edges to insert where
$n$ and $m$ are the number of nodes and edges in the network,
respectively. Here, the main motivation can be calibrating the
importance/load of some or all of the vertices as desired, matching
their loads to their capacities, boosting the content spread, or
making the network immune to adversarial attacks. Similar problems,
such as finding the most cost-effective way which reduces the entity
dissemination ability of a network~\cite{Phillips1993} or finding a
small set of edges whose deletion maximizes the shortest-path
length~\cite{Israeli02}, have been investigated in the literature. The
problem recently regained a lot of attention: A generic study which
uses edge insertions and deletions is done by
Tong~et~al.~\cite{Tong12}. They use the changes on the leading
eigenvalue to control/speed-up the dissemination process. Other recent
works investigate edge insertions to minimize the average shortest
path distance~\cite{Papagelis2011} or to boost the content
spread~\cite{Chaoji2012}. From the centrality point of view, there
exist studies which focus on maximizing the centrality of a node
set~\cite{Everett2005,Ishakian2012} or a single
node~\cite{Ishakian2012} by edge insertions. In generic
centrality-based network management problem, the desired centralities
of all the nodes need to be obtained or approximated with a small set
of topology modifications. As Figure~\ref{fig:gseq} shows, the effect
of a local topology modification is usually global. Furthermore,
existing algorithms for incremental centrality computation
are not efficient enough to be used in practice. Thus,
novel incremental algorithms are essential to quickly evaluate the
effects of topology modifications on centrality values.

\begin{figure}[htbp] 
\vspace*{-1ex}
\center
~~~\includegraphics[width=0.38\textwidth]{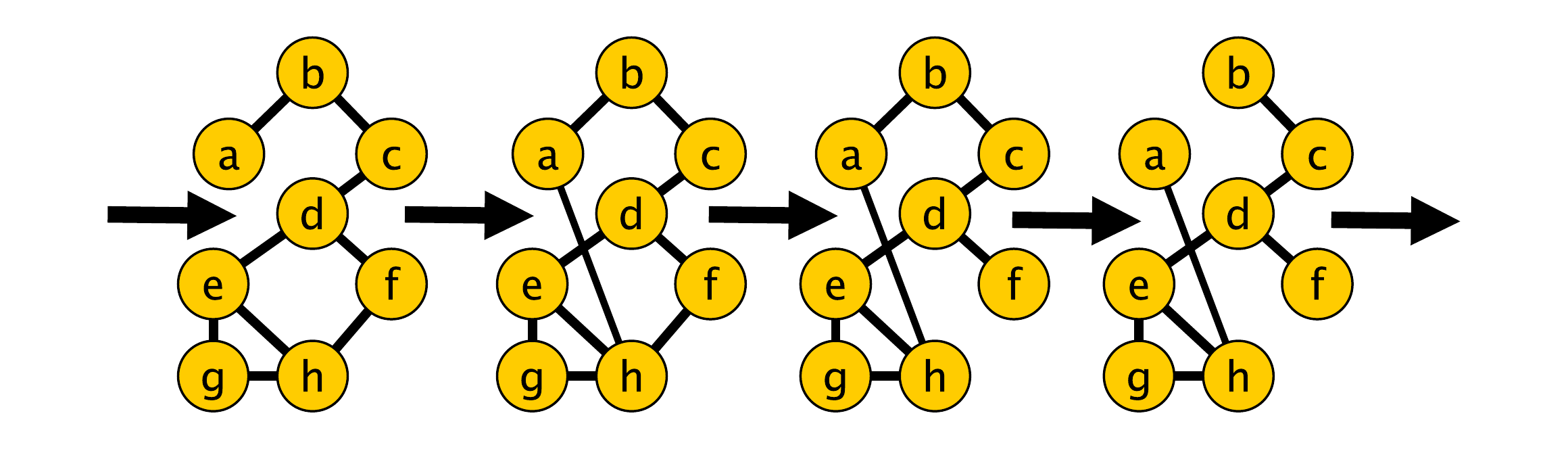}
\includegraphics[width=0.47\textwidth]{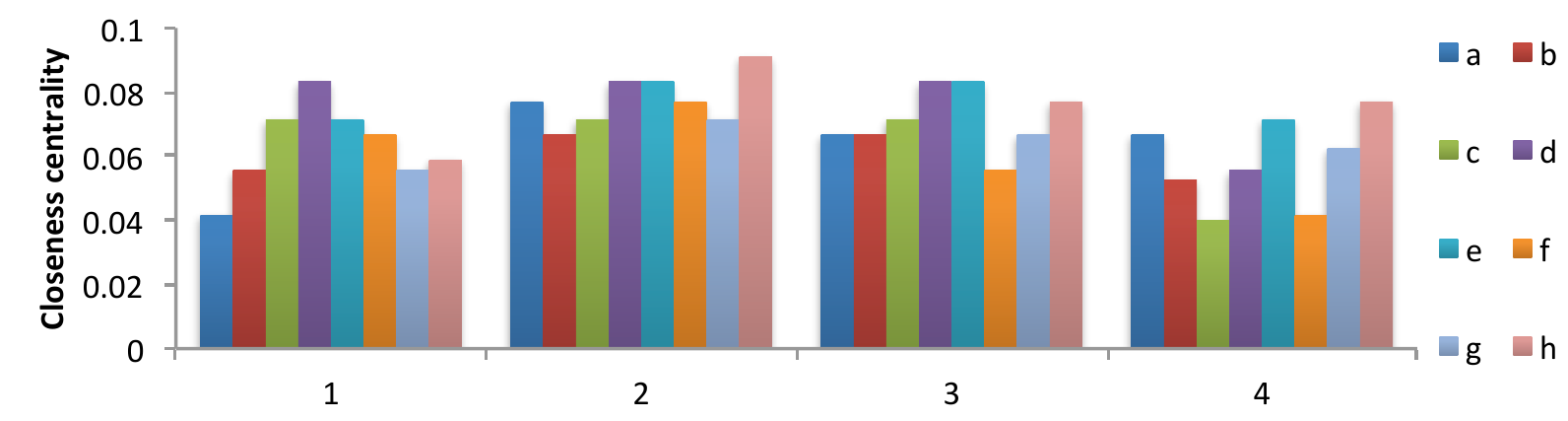}
\vspace*{-1ex}
\caption{\small A toy network with nine nodes, three consecutive edge~($ah$,
$fh$, and $ab$, respectively) insertions/deletions, and values of
closeness centrality.
}  \label{fig:gseq} 
\vspace*{-1ex}
\end{figure}

Our contributions can be summarized as follows:
\begin{enumerate}[topsep=4pt, itemsep=0pt]
\item{To attack the variants of the centrality-based network
  management problem, we propose incremental algorithms which
  efficiently update the closeness centralities upon edge insertions
  and deletions.}

\item{The proposed algorithms can serve as a fundamental building
  block for other shortest-path-based network analyses such as the
  temporal analysis on the past network data,
  maintaining centrality on streaming networks, or
  minimizing/maximizing the average shortest-path distance via edge
  insertions and deletions.}

\item{Compared with the existing algorithms, our algorithms have a
  low-memory footprint making them practical and applicable to very
  large graphs. For random edge insertions/deletions to the Wikipedia
  users' communication graph, we reduced the centrality
  (re)computation time from 2 days to 16 minutes. And for the
  real-life temporal DBLP coauthorship network, we reduced the time
  from 1.3 days to 4.2 minutes.}

\item{The proposed techniques can easily be adapted to algorithms for
  approximating centralities. As a result, one can employ a more
  accurate and faster sampling and obtain better approximations.}
\end{enumerate}

The rest of the paper is organized as follows: Section~\ref{sec:bac}
introduces the notation and formally defines the closeness centrality
metric. Section~\ref{sec:manage} defines network management problems
we are interested. Our algorithms explained in detail in
Section~\ref{sec:main}. Existing approaches are described in
Section~\ref{sec:rel} and the experimental analysis is given in
Section~\ref{sec:exp}. Section~\ref{sec:con} concludes the
paper.

\section{Background}\label{sec:bac}

Let $G = (V,E)$ be a network modeled as a simple graph with $n = |V|$
vertices and $m = |E|$ edges where each node is represented by a vertex in
$V$, and a node-node interaction is represented by an edge in $E$. 
Let $\Gamma_G(v)$ be the set of vertices which are connected to $v$ in
$G$.

A graph $G' = (V', E')$ is a {\em subgraph} of $G$ if $V' \subseteq V$
and $E' \subseteq E$. A \textit{path} is a sequence of vertices such
that there exists an edge between consecutive vertices. A path between
two vertices $s$ and $t$ is denoted by $s \leadsto t$~(we sometimes use $s
 \stackrel{P}{\leadsto} t$ to denote a specific path $P$ with
 endpoints $s$ and $t$). Two vertices
$u, v \in V$ are \textit{connected} if there is a path from $u$ to
$v$.  If all vertex pairs are connected we say that $G$ is
\textit{connected}. If $G$ is not connected, then it is {\em
disconnected} and each maximal connected subgraph of $G$ is a {\em
connected component}, or a component, of $G$. We use $\dis_G(u,v)$ to
denote the length of the shortest path between two vertices $u, v$ in
a graph $G$. If $u = v$ then $\dis_G(u,v) = 0$. And if $u$ and $v$ are
disconnected, then $\dis_G(u,v) = \infty$.
 
Given a graph $G = (V,E)$, a vertex $v \in V$ is called an {\em
  articulation vertex} if the graph $G-v$ (obtained by removing $v$)
has more connected components than $G$. Similarly, an edge $e \in E$
is called a {\em bridge} if $G-e$ (obtained by removing $e$ from $E$)
has more connected components than $G$. $G$ is {\em biconnected} if it
is connected and it does not contain an articulation vertex. A maximal
biconnected subgraph of $G$ is a {\em biconnected component}.

\subsection{Closeness Centrality}

Given a graph $G$, the {\em farness} of a vertex $u$ is
defined as $$\far[u] = \sum_{\stackrel{v \in V}{\dis_G(u,v) \neq
    \infty}} \dis_G(u,v).$$ And the
closeness centrality of $u$ is defined as
\begin{equation}
\cc[u] = \frac{1}{\far[u]}. \label{eq:first}
\end{equation}
If $u$ cannot reach any vertex in the graph $\cc[u] = 0$.

For a sparse unweighted graph $G \ (V,E)$ with $|V| = n$ and $|E| =
m$, the complexity of \cc computation is $\mathcal{O}(n(m+n))$.  For
each vertex $s$, Algorithm~\ref{alg:cc} executes a Single-Source Shortest
Paths~(SSSP) algorithm.  It initiates a breadth-first search~(BFS)
from $s$, computes the distances to the other vertices, compute
$\far[s]$, the sum of the distances which are different than $\infty$.
And, as the last step, it computes $\cc[s]$. Since a BFS takes $\mathcal{O}(m+n)$
time, and $n$ SSSPs are required in total, the complexity follows.

\renewcommand{\baselinestretch}{0.70}
\begin{algorithm}[htbp]
\DontPrintSemicolon
\SetKwComment{tcp}{$\triangleright$}{} \small
\caption{\ccalg: Basic centrality computation}
\label{alg:cc}
\KwData{${G = (V,E)}$}  
\KwOut{$cc[.]$}
  \lnl{line:bfs}\For{{\bf each} $s \in V$} {    
  	\tcp{SSSP($G$, $s$) with centrality computation}
    $Q \leftarrow$ empty queue \;    
    $\dis[v] \gets \infty, \forall v \in V\setminus{\{s\}}$ \;
    $Q$.push($s$),\ $\dis[s] \gets 0$ \; 
    $\far[s] \gets 0$ \;
    \While{$Q$ is not empty} {
      $v \leftarrow Q$.pop()\;
      \For{{\bf all} $w \in \Gamma(v)$}{
        \If{$\dis[w] = \infty$}{
          $Q$.push($w$) \;
          $\dis[w] \gets \dis[v] + 1$ \;
	  $\far[s] \gets \far[s] + \dis[w]$\;
        }
      }
    }
    $\cc[s] = \frac{1}{\far[s]}$
  }
  \Return{$\cc[.]$}\;
 \end{algorithm}
\renewcommand{\baselinestretch}{1}

\section{Problem Definitions}\label{sec:manage}
The following problem can be considered as a generalized version of
the problems investigated in~\cite{Everett2005,Ishakian2012}.

\begin{definition} \emph{(Centrality-based network management)}
  Let $G = (V,E)$ be a graph. Given a centrality metric $\cal C$, a
  target centrality vector $\cv'$, and an upper bound $U$ on the
  number of inserted/deleted edges, construct a graph $G' = (V, E')$,
  s.t., $|E \Delta E'| \leq U$ and $||\cv' - \cv_{G'}||$ is minimized.
\label{def:gencen}
\end{definition}

In this work, we are interested in the closeness metric which is based
on shortest paths. Hence, implicitly, we are also interested in the
following problem partly investigated
in~\cite{Israeli02,Papagelis2011,Phillips1993}.

\begin{definition} \emph{(Shortest-path-based network management)}
  Let $G = (V,E)$ be a graph. Given an upper bound $U$ on the number
  of inserted/deleted edges, construct a graph $G' = (V, E')$ where
  $|E \Delta E'| \leq U$ and the (average) shortest-path in $G'$ is
  minimized/maximized.
\label{def:genshor}
\end{definition}

These problems and their variants have several applications such as
slowing down pathogen outbreaks, increasing the efficiency of the
advertisements, and analyzing the robustness of a network. Consider an
airline company with flights to thousands of airports and aim to add
some new routes to increase the load of some underutilized
airports. When a new route is inserted, in order to evaluate its
overall impact, all the airport centralities need to be
re-computed which is a quite expensive task. Hence, we need to have
efficient incremental algorithms to tackle this problem. Such
algorithms can be used as a fundamental building block to centrality-
and shortest-path-based network management problems (and their variants)
as well as temporal centrality/shortest-path analyses and dynamic
network analyses. In this work, we investigate this subproblem.

\begin{definition} \emph{(Incremental closeness centrality)}
  Given a graph $G = (V,E)$, its centrality vector $\cc$, and an edge
  $uv$, find the centrality vector $\cc'$ of the graph $G' = (V,E
  \cup \{uv\})$~(or $G' = (V,E \setminus \{uv\})$).
\end {definition}



\section{Maintaining Centrality}\label{sec:main}

Many interesting real-life networks are scale free. The diameters of
these networks grow proportional to the logarithm of the number of
nodes. That is, even with hundreds of millions of vertices, the
diameter is small, and when the graph is modified with minor updates,
it tends to stay small. Combining this with their power-law degree
distribution, we obtain the spike-shaped shortest-distance
distribution as shown in Figure~\ref{fig:levels}. We use two main
approaches: {\em work filtering} and {\em SSSP hybridization} to
exploit these observations and reduce the centrality computation
time.

\subsection{Work Filtering}

For efficient maintenance of closeness centrality in case of an edge
insertion/deletion, we propose a {\em work filter} which reduces the
number of SSSPs in Algorithm~\ref{alg:cc} and the cost of each
SSSP. Work filtering uses three techniques: filtering with {\em level
differences}, with {\em biconnected component decomposition}, and with
{\em identical vertices}.

\subsubsection{Filtering with level differences}

The motivation of level-based filtering is detecting the unnecessary
updates and filtering them.  Let $G = (V, E)$ be the current graph and
$uv$ be an edge to be inserted to $G$. Let $G' = (V, E \union {uv})$
be the updated graph. The centrality definition in~\eqref{eq:first}
implies that for a vertex $s \in V$, if $\dis_G(s,t) = \dis_{G'}(s,t)$
for all $t \in V$ then $\cc[s] = \cc'[s]$. The following theorem is
used to detect such vertices and filter their SSSPs.

\begin{theorem}
Let $G = (V,E)$ be a graph and $u$ and $v$ be two vertices in $V$
s.t. $uv \notin E$. Let $G' = (V,E \union {uv})$. Then $\cc[s] =
\cc'[s]$ if and only if $|\dis_G(s,u) - \dis_G(s,v)|
\leq 1$.\label{thm:add}
\end{theorem}

\begin{proof}
If $s$ is disconnected from $u$ and $v$, $uv$'s insertion will not
change the closeness centrality of $s$. Hence, $\cc[s] = \cc'[s]$. If $s$
is only connected to one of $u$ and $v$ in $G$ the difference $|\dis_G(s,u) -
\dis_G(s,v)|$ is $\infty$, and the closeness centrality score of $s$ needs
to be updated by using the new, larger connected component containing
$s$.

When $s$ is connected to both $u$ and $v$ in $G$, we investigate the edge
insertion in three cases as shown in Figure~\ref{fig:leveldiffs}:

Case 1. $\dis_G(s,u) = \dis_G(s,v)$: Assume that the path $s
  \stackrel{P}{\leadsto} u$--$v \stackrel{P'}{\leadsto} t$ is a
  shortest $s \leadsto t$ path in $G'$ containing $uv$. Since
  $\dis_G(s,u) = \dis_G(s,v)$ there exist another path $s
  \stackrel{P''}{\leadsto} v \stackrel{P'}{\leadsto} t$ in $G'$ with
  one less edge. Hence, $uv$ cannot be in a shortest path: 
  $\forall t \in V, \dis_G(s,t) = \dis_{G'}(s,t)$.

Case 2. $|\dis_G(s,u) - \dis_G(s,v)| = 1$: Let $\dis_G(s,u) <
  \dis_G(s,v)$ and assume that $s \stackrel{P}{\leadsto} u$--$v
  \stackrel{P'}{\leadsto} t$ is a shortest path in $G'$ containing
  $uv$. Since $\dis_G(s,v) = \dis_G(s,u) + 1$, there exist another
  path $s \stackrel{P''}{\leadsto} v \stackrel{P'}{\leadsto} t$ in
  $G'$ with the same number of edges. Hence, $\forall t \in V, \dis_G(s,t) =
  \dis_{G'}(s,t)$.

Case 3. $|\dis_G(s,u) - \dis_G(s,v)| > 1$: Let $\dis_G(s,u) <
  \dis_G(s,v)$. The path $s \leadsto u$--$v$ in $G'$ is shorter than
  the shortest $s \leadsto v$ path in $G$ since $\dis_G(s,v) > \dis_G(s,u)
  + 1$. Hence, an update on $\cc[s]$ is necessary.
\end{proof}

\begin{figure}[htbp]
\vspace*{-1ex}
\center
\includegraphics[width=0.45\textwidth]{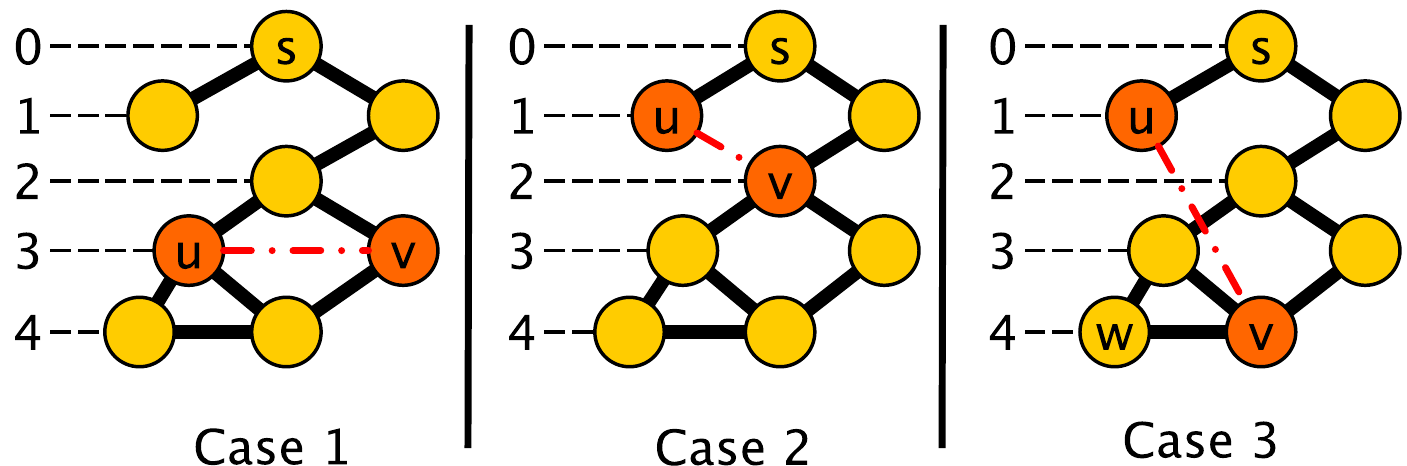}
\vspace*{-1ex}
\caption{\small (1) Three cases of edge insertion: when an edge $uv$ is inserted to the
 graph $G$, for each vertex $s$, one of them is true:
 (a) $\dis_G(s,u) = \dis_G(s,v)$, (b) $|\dis_G(s,u) - \dis_G(s,v)| =
 1$, and (c) $|\dis_G(s,u) - \dis_G(s,v)| > 1$.}
\label{fig:leveldiffs}
\vspace*{-1ex}
\end{figure}

Although Theorem~\ref{thm:add} yields to a filter only in case 
of edge insertions, the following corollary which is used for edge 
deletion easily follows.
\begin{corollary}
Let $G = (V, E)$ be a graph and $u$ and $v$ be two vertices in $V$
s.t. $uv \in E$. Let $G' = (V,E \setminus \{uv\})$. Then $\cc[s] =
\cc'[s]$ if and only if $|\dis_{G'}(s,u) - \dis_{G'}(s,v)|
\leq 1$.
\end{corollary}
With this corollary, the work filter can be implemented for both 
edge insertions and deletions. 
The pseudocode of the update algorithm in case of
an edge insertion is given in Algorithm~\ref{alg:filtered}.
When an edge $uv$ is inserted/deleted, to employ the filter, we first compute
the distances from $u$ and $v$ to all other vertices. And, it filters
the vertices satisfying the statement of Theorem~\ref{thm:add}.

\renewcommand{\baselinestretch}{0.70}
\begin{algorithm}[htbp]
\DontPrintSemicolon
\SetKwComment{tcp}{$\triangleright$}{} \small
\caption{Simple work filtering}
\label{alg:filtered}
\KwData{${G = (V,E)}$, $\cc[.]$, $uv$} 
\KwOut{$\cc'[.]$} 
  $G' \gets (V, E \union \{uv\})$\;
  $\dis u[.] \gets $ SSSP($G$, $u$) \tcp{\ distances from $u$ in $G$}
  $\dis v[.] \gets $ SSSP($G$, $v$) \tcp{\ distances from $v$ in $G$}
  
  \For{{\bf each} $s \in V$} {
	\If{$|\dis u[s] - \dis v[s]| \leq 1$} {
		$\cc'[s] = \cc[s]$\;
     } \Else {
     	\tcp{\ use the computation in Algorithm~\ref{alg:cc} with $G'$}
     }
  }
  \Return{$\cc'[.]$}\;
 \end{algorithm}
\renewcommand{\baselinestretch}{1}

In theory, filtering by levels can reduce the update time
significantly. However, in practice, its effectiveness depends on the
underlying structure of $G$. Many real-life networks have been
repeatedly shown to possess unique characteristics such as a small
diameter and a power-law degree distribution~\cite{McGlohon11}.  And
the spread of information is extremely
fast~\cite{Doerr11,Doerr12}. The proposed filter exploits one of these
characteristics for efficient closeness centrality updates: the
distribution of shortest-path lengths. Its efficiency is based on the
phenomenon shown in Figure~\ref{fig:levels} for a set of graphs used
in our experiments: the probability distribution function for a
shortest-path length being equal to $x$ is unimodular and spike-shaped
for many social networks and also some others. This is the outcome of
the short diameter and power-law degree distribution. On the other
hand, for some spatial networks such as road networks, there are no
sharp peaks and the shortest-path distances are distributed in a more
uniform way. The work filter we propose here prefer the former.

\begin{figure}[htbp]
\vspace*{-1ex}
\center
\includegraphics[width=0.9\columnwidth]{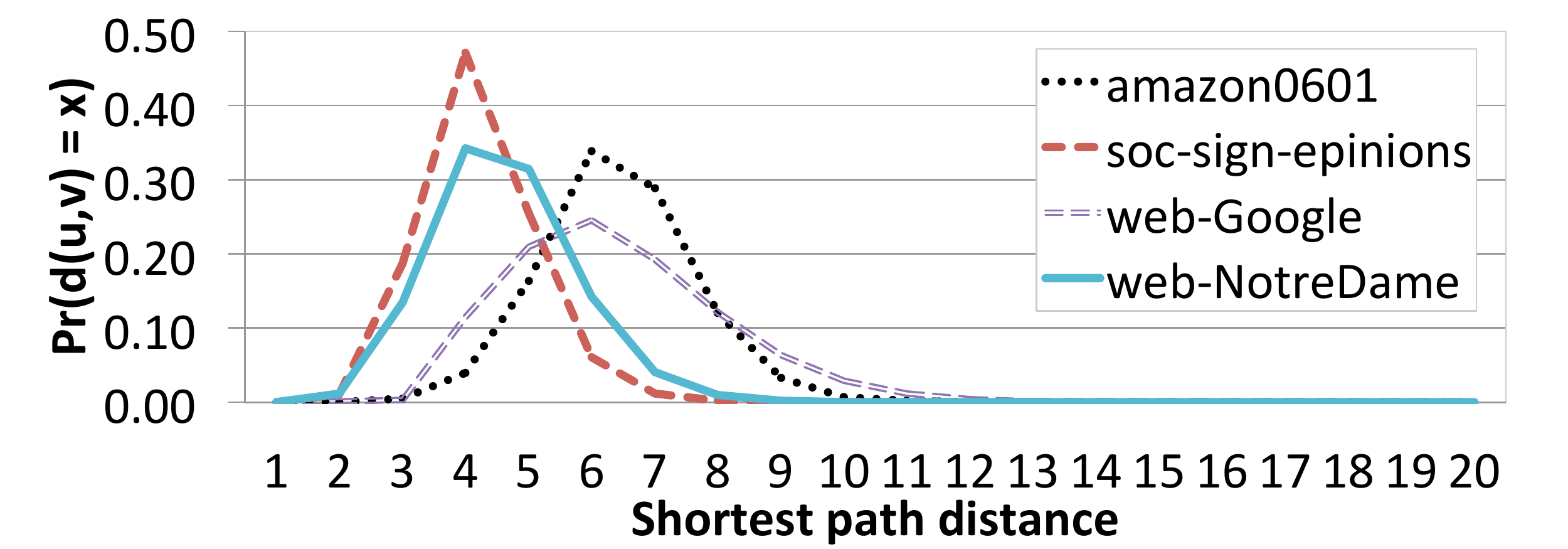}
\vspace*{-1ex}
\caption{\small Probability of the distance between two~(connected) vertices
is equal to $x$ for four social and web networks.}
\label{fig:levels}
\vspace*{-1ex}
\end{figure}

\subsubsection{Filtering with biconnected components}

Our work filter can be enhanced by employing and maintaining a
biconnected component decomposition~(BCD) of $G = (V,E)$.  A BCD is a
partitioning $\Pi$ of the edge set $E$ where $\Pi(e)$ indicates the
component of each edge $e \in E$. A toy graph and its BCDs before and
after edge insertions are given in Figure~\ref{fig:bcd}.

\begin{figure}[htbp]
\vspace*{-1ex}
\center
\hspace*{-1ex}
\subfigure[$G$]{\includegraphics[height=1.7cm]
{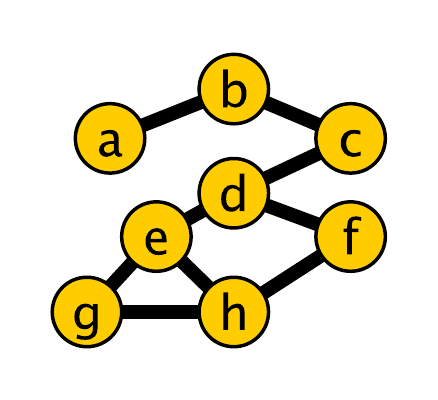}\label{fig:graph}}
\hspace*{1ex}
\subfigure[$\Pi$]{\includegraphics[height=1.85cm]
{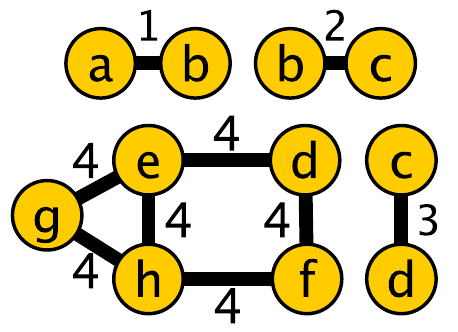}\label{fig:pig}}
\hspace*{1ex}
\subfigure[$\Pi'$]{\includegraphics[height=1.9cm]
{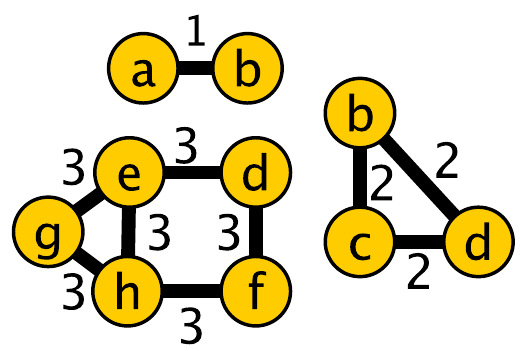}\label{fig:pig2}}
\vspace*{1ex}
\caption{\small A graph $G$~(left), its biconnected component decomposition
$\Pi$ into $4$ components~(middle), and the updated $\Pi'$ with $3$
components when the edge $bd$ is inserted~(right). The sets of
articulation vertices before and after the edge insertion are
$\{b,c,d\}$ and $\{b,d\}$, respectively. After the edge addition, $cid
= 2$. That is to say, the second component contains the new
edge. Hence, the biconnected component $2$ is extracted first and
executes an update algorithm only for the vertices $\{b,c,d\}$. It
also initiates a fixing phase to update the closeness centrality
values for the rest of the vertices. After the edge insertion,
$\rep[a] = b$, and $\rep[e] = \rep[f] = \rep[g] = \rep[h] = b$. Hence,
$\card[b] = 2$, $\card[c] = 1$, and $\card[d] = 5$. And, $\sumcard[b]
= 1$, $\sumcard[c] = 0$, and $\sumcard[d] = 6$.}
\label{fig:bcd}
\vspace*{-1ex}
\end{figure}

When $uv$ is inserted to $G = (V,E)$ and $G' = (V,E' \union \{uv\})$
is obtained, we check if
$$\{\Pi(uw): w \in \Gamma_G(u)\} \cap \{\Pi(vw): w \in \Gamma_G(v)\}$$
is empty or not. If the intersection is not empty, there will be only
one element in it, $cid$, which is the id of the biconnected component
of $G'$ containing $uv$~(otherwise $\Pi$ is not a valid BCD). In this case, $\Pi'(e)$ is set to $\Pi(e)$ for all $e \in
E$ and $\Pi'(uv)$ is set to $cid$. If there is no biconnected
component containing both $u$ and $v$~(see Figure~\ref{fig:pig2}), i.e.,
if the intersection above is empty, we construct $\Pi'$ from
scratch and set $cid = \Pi'(uv)$. $\Pi$ can be computed in linear, $\mathcal{O} (m +
n)$ time~\cite{Hopcroft10}. Hence, the
cost of BCD maintenance is negligible compared to the cost of
updating closeness centrality.
  
Let $G'_{cid} = (V_{cid}, E'_{cid}) $ be the biconnected component of $G'$ containing $uv$ where 
\begin{align*}
V_{cid} &= \{v \in V: cid \in \{\Pi'(vw): vw \in E'\}\}, \\
E'_{cid} &= \{e \in E': \Pi'(e) = cid\}.
\end{align*}
Let $\Am_{cid} \subseteq V_{cid}$ be the set of articulation vertices in
$G'_{cid}$. Given $\Pi'$, it is easy to detect the articulation
vertices since  $u$ is an articulation vertex if and only if it is
part of at least
two components in the BCD: $|\{\Pi'(uw): uw \in E'\}| > 1$.

We will execute SSSPs only for the vertices in $G'_{cid}$ and use
the new values to fix the centralities for the rest of the graph. The
contributions of the vertices in $V \setminus V_{cid}$ are integrated
to the SSSPs by using a representative function
$\rep: V \rightarrow V_{cid} \union \{\tt{null}\}$ which maps each
vertex $v \in V$ either to a representative in $G'_{cid}$ or to
$\tt{null}$ (if $v$ and the vertices in $V_{cid}$
are in different connected components of $G'$).

For each vertex $u \in V_{cid}$, we set $\rep[u] = u$.  For the other
vertices, let $\overline{G'_{cid}} = \{V, E' \setminus E'_{cid}\}$. If
a vertex $v \in V \setminus V_{cid}$ and an articulation vertex $u \in
\Am_{cid}$ are connected in $\overline{G'_{cid}}$, i.e.,
$\dis_{\overline{G'_{cid}}}(u,v) \neq \infty$, we say that $v$ is
represented by $u$ in $G'_{cid}$ and set $\rep[v] = u$. Otherwise,
$\rep[v]$ is set to ${\tt null}$. The following theorem states that
$\rep$ is well defined: each vertex is represented by at most one vertex.

\begin{theorem}\label{thm:rep}
 For each $v$ in $V \setminus V_{cid}$, there is at most
 one articulation vertex $u \in \Am_{cid}$ such that
 $\dis_{\overline{G'_{cid}}}(u,v) \neq \infty$.
\end{theorem}
\begin{proof}
The proof directly follows from the definition of BCD and is omitted.
\end{proof}

Since all the~(shortest) paths from a vertex $v \in V\setminus V_{cid}$ to a
vertex in $V_{cid}$ are passing through $\rep[v]$, the following is a
corollary of the theorem.
\begin{corollary}\label{cor:path}
For each vertex $v \in V \setminus V_{cid}$ with $\rep[v] \neq {\tt
null}$,
$\dis_{\overline{G'_{cid}}}(v, \rep[v]) = \dis_{G'}(v, \rep[v]),$
which is different than $\infty$.  Furthermore, for a vertex $w \in V$
which is also represented in $G'_{cid}$ but not in the connected
component of ${\overline{G'_{cid}}}$ containing $v$, $\dis_{G'}(v,w)$
is equal to 
$$\dis_{G'}(v,\rep[v]) + \dis_{G'}(\rep[v],\rep[w]) + \dis_{G'}(\rep[w],w).$$
If $w \in V_{cid}$ the last term on the right is $0$, since $\rep[w] =
w$.
\end{corollary}

To correctly update the new centrality values, we compute two
extra values for each vertex $u \in V_{cid}$,
\begin{align}
\card[u] &= \left|\{v \in V: \rep[v] = u\}\right|,\label{eq:card}\\
\sumcard[u] &= \sum_{\stackrel{v \in V}{\rep[v] = u}}\dis_{G'}(u,v).\label{eq:sumcard}
\end{align}
That is, $\card[u]$ is the number of vertices in $V$ which are
represented by $u$~(including $u$). And $\sumcard[u]$ is the farness of
$u$ to these vertices in $G'$. The modified update
algorithm is given in Algorithm~\ref{alg:combined}.

\renewcommand{\baselinestretch}{0.70}
\begin{algorithm}[h]
\DontPrintSemicolon
\SetKwComment{tcp}{$\triangleright$}{} \small
\caption{Update with BCD and level differences}
\label{alg:combined}
\KwData{${G = (V,E)}$, $\Pi$, $\cc$, $\far$, $uv$} 
\KwOut{$\cc'[.]$, $\far'[.]$}
  \lnl{line:filter}\tcp{\ prepare for filtering} 
  $G' \gets (V, E')$ where $E' \gets E \union \{uv\}$\;
  $cSet_u \gets \{\Pi(uw): w \in \Gamma_G(u)\}$\;
  $cSet_v \gets \{\Pi(vw): w \in \Gamma_G(v)\}$\;

  \If{$cSet_u \cap cSet_v \neq \emptyset$} {
    $cid \gets $ \#the common component\;
    $\Pi'(e) \gets \Pi(e)$ $\forall e \in E$, $\Pi'(uv) \gets cid$\;
  } \lElse {
    construct $\Pi'$ from $G'$, $cid \gets \Pi'(uv)$\;
  }
  
  $V_{cid} \gets \{v \in V: cid \in \{\Pi'(vw): vw \in E'\}\}$ \;
  $E'_{cid} \gets \{e \in E': \Pi'(e) = cid\}$ \;
  $G'_{cid} = (V_{cid}, E'_{cid})$ \;
  $G_{cid} = (V_{cid}, E'_{cid} \setminus \{uv\})$\;

  Set $\rep[v]$, $\forall v \in V$\;
  $\card[u] \gets \left|\{v \in V, \rep[v] = u\}\right|$, $\forall u \in V_{cid}$\;
  $\sumcard[u] \gets \sum_{v \in V, \rep[v] = u}\dis_{G'}(u,v)$, $\forall u \in V_{cid}$\;

  $\dis u[.] \gets $ SSSP($G_{cid}$, $u$), $\dis v[.] \gets $ SSSP($G_{cid}$, $v$)\;
 
  \lnl{line:update}\tcp{\ update phase}
  \For{{\bf each} $s \in V_{cid}$} {
	\lIf{$|\dis u[s] - \dis v[s]| \leq 1$} {
	  $\cc'[s] = \cc[s]$\;
     } \Else {
	\lnl{line:source}$Q \leftarrow$ empty queue \;    
	  $\dis[v] \gets \infty$, $\forall v \in V_{cid} \setminus{\{s\}}$ \;
	  $Q$.push($s$),\ $\dis[s] \gets 0$ \; 
	  $\far'[s] \gets 0$ \;
	  \While{$Q$ is not empty} {
	    $v \leftarrow Q$.pop()\;
	    \For{{\bf all} $w \in \Gamma_{G'_{cid}}(v)$}{
              \If{$\dis[w] = \infty$}{
		$Q$.push($w$) \;
		$\dis[w] \gets \dis[v] + 1$ \;
		\lnl{line:far}$\far'[s] \gets \far'[s] + (\dis[w] \times \card[w]) + \sumcard[w]$\;
              }
	    }
	  }
	  $\cc'[s] = \frac{1}{\far'[s]}$	  
	}
  }

  \lnl{line:fix}\tcp{\ fix phase}
  \For{{\bf each} $v \in V \setminus V_{cid}$} {
    $r \gets \rep[v]$\;
    \If{$r \neq {\tt{null}}$ {\bf and} $\far[r] \neq \far'[r]$} {      
      \lnl{line:farup}$\far'[v] \gets \far[v] - (\far[r] - \far'[r])$\;
      $\cc'[v] \gets \frac{1}{\far'[v]}$\;     
    }
  }

  \Return{$\cc'[.]$}\;
 \end{algorithm}
\renewcommand{\baselinestretch}{1}

\begin{lemma}\label{lem:bcd1}
For each vertex $v \in V_{cid}$, Algorithm~\ref{alg:combined} computes
the correct $\cc'[v]$ value. 
\end{lemma}

\begin{proof}
We will prove that $\far'[v]$ is correct for all $v \in V_{cid}$. Let
$v = s$ be the vertex whose closeness centrality update is started at
line~\ref{line:source}. At line~\ref{line:far} of
Algorithm~\ref{alg:combined}, the update on $\far'[v]$ is
$\dis_{G'}(v,w) \times \card[w] + \sumcard[w]$
which can be rewritten as 
$$\sum_{\stackrel{u \in V}{\rep[u] =  w}}\dis_{G'}(v,w) + \dis_{G'}(w,u),$$
by using~\eqref{eq:card} and~\eqref{eq:sumcard}. According to
Corollary~\ref{cor:path}, this is equal to 
$$\sum_{\stackrel{u \in V}{\rep[u] = w}}\dis_{G'}(v,u).$$
Due to the definition of $\rep$, only the
vertices which are connected to $v$ will have an effect on $\far'[v]$.
And due to Theorem~\ref{thm:rep}, each vertex can contribute to
at most one update. Hence 
$$\sum_{w \in V_{cid}}\sum_{\stackrel{u \in V} {\rep[u] =
w}}\dis_{G'}(v,u) = \sum_{\stackrel{u \in V}{\dis_{G'}(v,u) \neq
\infty}}\dis_{G'}(v,u),$$ which is the $\far'[v]$ in $G'$ as
desired.
\end{proof}

\begin{lemma}\label{lem:bcd2}
For each vertex $v \in V \setminus V_{cid}$, Algorithm~\ref{alg:combined} computes
the correct $\cc'[v]$ value. 
\end{lemma}

\begin{proof}
We will prove that $\far'[v]$ is correct for all $v \in V \setminus
V_{cid}$ after the fix phase. Let $u = \rep[v]$. If $u$ is {\tt
null} then $v$'s farness and hence closeness value will remain the
same.

Assume that $u$ is not null. Let $w$ be a vertex with $\rep[w] \neq
{\tt null}$. If $w$ and $v$ are in the same connected component of
$\overline{G'_{cid}}$ then $\dis_G(v,w) = \dis_{G'}(v,w)$ and
$\dis_G(u,w) = \dis_{G'}(u,w)$. Hence, the change on $\far[v]$ and
$\far[u]$ due to $w$ are both $0$. On the other hand, if $w$
is in a different connected component of $\overline{G'_{cid}}$
according to Corollary~\ref{cor:path},
$$\dis_{G'}(v,w) = \dis_{G'}(v,u) + \dis_{G'}(u,\rep[w]) + \dis_{G'}(\rep[w],w),$$
where the sum of the second and the third terms is equal
to $\dis_{G'}(u,w)$. Since the first term does not change by the
insertion of $uv$, the change on $\dis_{G'}(u,w)$ is equal to the
change on $\dis_{G'}(v,w)$. That is when aggregated, the change on
$\far[v]$ is equal to the change on $\far[u]$. Lemma~\ref{lem:bcd1}
implies that $\far[u]$ is correct. Hence, $\far'[v]$, computed at
line~\ref{line:farup}, must also be correct.
\end{proof}

\begin{theorem}
For each vertex $v \in V$, Algorithm~\ref{alg:combined} computes
the correct $\cc'[v]$ value. 
\end{theorem}
\begin{proof}
Follows from Lemma~\ref{lem:bcd1} and~\ref{lem:bcd2}.
\end{proof}

The complexity of the update algorithm is $\mathcal{O}(n(m+n))$. And
the overhead of filter preparation~(line~\ref{line:filter}
through~\ref{line:update}) is $\mathcal{O}(m+n)$ since it only
contains a constant number of graph traversals. In case of an edge
deletion, it is enough to get $G'_{cid}$ as the biconnected component
which was containing the deleted edge. The rest of the procedure can 
be adapted in a straightforward manner. 

\subsubsection{Filtering with identical vertices}

Our preliminary analyses on various networks show that some of the
graphs contain a significant amount of {\em identical} vertices which
have the same/a similar neighborhood structure. This can be exploited
to reduce the number of SSSPs further. We investigate two types of
identical vertices.

\begin{definition}
In a graph $G$, two vertices $u$ and $v$ are \em{type-I-identical} if and only if
$\Gamma_G(u) = \Gamma_G(v)$.
\end{definition}

\begin{definition}
In a graph $G$, two vertices $u$ and $v$ are \em{type-II-identical} if
and only if $\{u\} \cup \Gamma_G(u) = \{v\} \cup \Gamma_G(v)$.
\end{definition}

Both types form an equivalance
class relation since they are reflexive, symmetric, and
transitive. Furthermore, all the non-trivial classes they form (i.e.,
the ones containing more than one vertex) are
disjoint. 

Let $u, v \in V$ be two identical vertices. One can see that for any
vertex $w \in V \setminus \{u,v\}$, $\dis_{G}(u,w) =
\dis_{G}(v,w)$. Then the following is true.

\begin{corollary}
Let ${\cal I} \subseteq V$ be a vertex-class containing type-I or type-II
identical vertices. Then the closeness centrality values of all the
vertices in ${\cal I}$ are equal.
\end{corollary}

To construct these equivalance classes for the initial graph, we first
use a hash function to map each vertex neighborhood to an integer:
$hash_I[u] = \sum_{v \in \Gamma_G(u)} v$.
We then sort the vertices with respect to their hash values and
construct the type-I vertex-classes by eliminating false positives due
to collisions on the hash function. 
A similar process is applied to detect type-II vertex classes. 
The complexity of this initial construction is
$\mathcal{O}(n \log n + m)$ assuming the number of collisions is small
and hence, false-positive detection cost is negligible.

Maintaining the equivalance classes in case of edge insertions and
deletions is easy: For example, when $uv$ is added to $G$, we first
subtract $u$ and $v$ from their classes and insert them to new
ones~(or leave them as singleton if none of the vertices are now
identical with them). The cost of this maintenance is $\mathcal{O}(n +
m)$. 

While updating closeness centralities of the vertices in $V$, we
execute an SSSP at line~\ref{line:source} of
Algorithm~\ref{alg:combined} for at most one vertex from each
class. For the rest of the vertices, we use the same closeness
centrality value. The improvement is straightforward and the
modifications are minor. For brevity, we do not give the
pseudocode.

\subsection{SSSP Hybridization}

The spike-shaped distribution given in Figure~\ref{fig:levels} can also
be exploited for SSSP hybridization. Consider the execution of
Algorithm~\ref{alg:cc}: while executing an SSSP with source $s$, for
each vertex pair $u,v$, $u$ is processed before $v$ if and only if
$d_G(s,u) < d_G(s,v)$. That is, Algorithm~\ref{alg:cc} consecutively
uses the vertices with distance $k$ to find the vertices with distance
$k+1$. Hence, it visits the vertices in a {\em top-down} manner.  SSSP
can also be performed in a a {\em bottom-up} manner. That is to say, after all
distance~(level) $k$ vertices are found, the vertices whose levels are
unknown can be processed to see if they have a neighbor at level
$k$. 


Figure~\ref{fig:bfstimes} gives the execution times of bottom-up and
top-down SSSP variants for processing each level. The trend for
top-down resembles the shortest distance distribution in small-world
networks. This is expected since in each level $\ell$, the vertices
that are $\ell$ step far away from $s$ are processed. On the other
hand, for the bottom-up variant, the execution time is decreasing
since the number of unprocessed nodes is decreasing. Following the
idea of Beamer~et~al.~\cite{Beamer12}, we hybridize the SSSPs
throughout the centrality update phase in
Algorithm~\ref{alg:combined}. We simply compare the number of edges
need to be processed for each variant and choose the cheaper one. For
the case presented in Figure~\ref{fig:bfstimes}, the hybrid algorithm
is $3.6$ times faster than the top-down variant.

\begin{figure}[htbp]
\vspace*{-1ex}
\center
\includegraphics[width=0.40\textwidth]{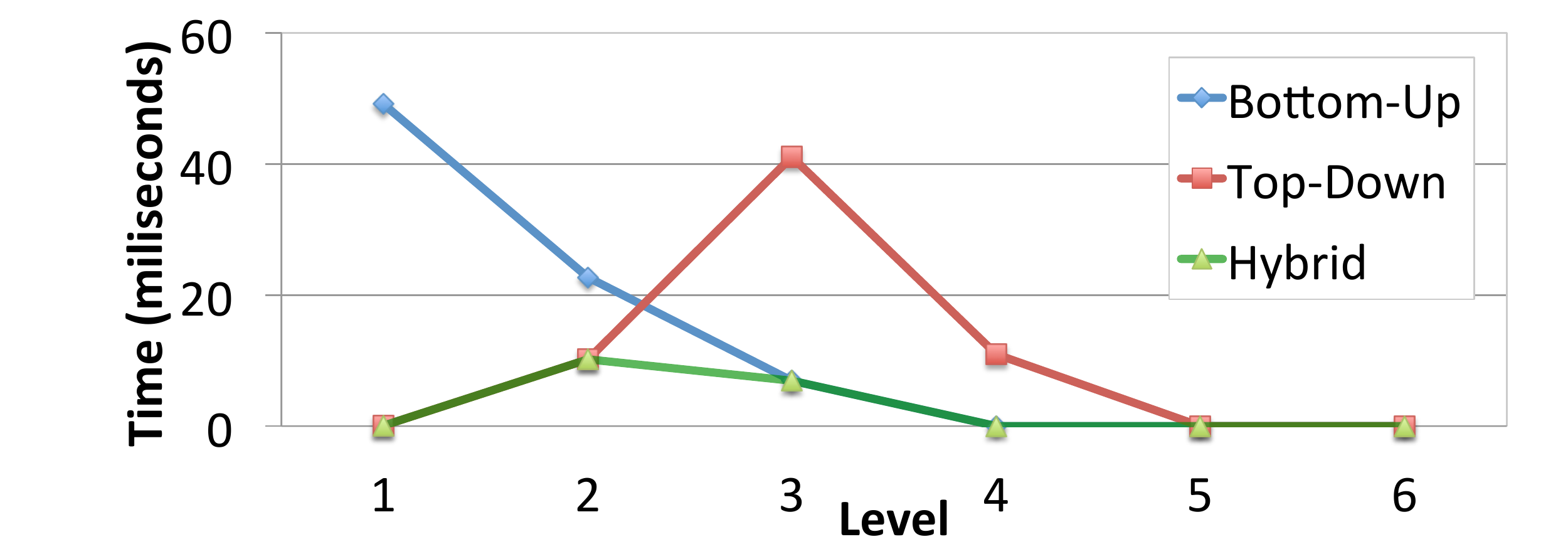}
\vspace*{-1ex}
\caption{\small Execution times of bottom-up, top-down, and hybrid SSSPs at
  each level for the {\it wiki-Talk} graph. The hybrid version is
  $3.63$ and $4.59$ times faster than the top-down and bottom-up
  versions respectively.
}
\label{fig:bfstimes}
\vspace*{-1ex}
\end{figure}

\section{Related Work}\label{sec:rel}

To the best of our knowledge, there are only two works that deal with
maintaining centrality in dynamic networks. Yet, both are interested
in betweenness centrality. Lee~et~al. proposed the {\bf QUBE}
framework which updates betweenness centrality in case of edge
insertion and deletion within the network~\cite{Lee2012}. {\bf QUBE}
relies on the biconnected component decomposition of the graphs. Upon
an edge insertion or deletion, assuming that the decomposition does not
change, only the centrality values within the updated biconnected
component are recomputed from scratch. If the edge insertion/deletion
affects the decomposition the modified graph is decomposed into its
biconnected components and the centrality values in the affected part
are recomputed. The distribution of the vertices to the biconnected
components is an important criteria for the performance of {\bf
QUBE}. If a large component exists, which is the case for many
real-life networks, one should not expect a significant reduction on
update time. Unfortunately, the performance of {\bf QUBE} is only
reported on small graphs~(less than 100K edges) with very low edge
density. In other words, it only performs significantly well on small
graphs with a tree-like structure having many small biconnected
components.

Green~et~al. proposed a technique to update centrality scores rather
than recomputing them from scratch upon edge insertions~(can be
extended to edge deletions)~\cite{Green2012}. The idea is storing the
whole data structure used by the previous betweenness centrality
update kernel.  This storage is indeed useful for two main reasons: it
avoids a significant amount of recomputation since some of the
centrality values will stay the same. And second, it enables a partial
traversal of the graph even when an update is necessary.  However, as
the authors state, $\mathcal{O}(n^2 + nm)$ values must be kept on the
disk. For the Wikipedia user communication and DBLP
coauthorship networks, which contain
thousands of vertices and millions of edges, the technique by
Green~et~al. requires TeraBytes of memory. The largest graph used
in~\cite{Green2012} has approximately $20K$ vertices and $200K$ edges;
the quadratic storage cost prevents their storage-based techniques to
scale any higher. On the other hand, the memory footprint of our
algorithms are linear and hence they are much more practical.

\section{Experimental Results}\label{sec:exp}

We implemented our algorithms in {\tt C}. The code is compiled with {\tt gcc
v4.6.2} and optimization flags {\tt -O2 -DNDEBUG}. The graphs are kept
in memory in the compressed row storage~(CRS) format. The experiments
are run on a computer with two Intel Xeon E$5520$ CPU clocked at $2.27$GHz
and equipped with $48$GB of main memory. All the experiments are run sequentially.

For the experiments, we used $10$ networks from the UFL Sparse Matrix
Collection\footnote{\url{http://www.cise.ufl.edu/research/sparse/matrices/}} and we also
extracted the coauthor network from current set of DBLP papers. 
Properties of the graphs are summarized in
Table~\ref{tab:graph_prop}. 
We symmetrized the directed graphs. The graphs are
listed by increasing number of edges and a distinction is made
between small graphs (with less than 500K edges) and the large graphs
(with more than 500K) edges.

\begin{table}
\smaller
\scalebox{0.85}{
 \begin{tabular}{|l|rr|rr|r|}
 \hline
 \multicolumn{3}{|c|}{\bf Graph}&\multicolumn{3}{c|}{{\bf Time} (in sec.)}\\ \hline
   name & {$|V|$} & {$|E|$} & Org. & Best & Speedup\\
 \hline
{\it hep-th} & 8.3K & 15.7K & 1.41 & 0.05 & 29.4 \\
{\it PGPgiantcompo} & 10.6K & 24.3K & 4.96 & 0.04 & 111.2 \\
{\it astro-ph} & 16.7K & 121.2K & 14.56 & 0.36 & 40.5 \\
{\it cond-mat-2005} & 40.4K & 175.6K & 77.90 & 2.87 & 27.2 \\ \hline
\multicolumn{5}{|r|}{\bf geometric mean}&{\bf 43.5}\\\hline
{\it soc-sign-epinions} & 131K & 711K & 778 & 6.25 & 124.5 \\
{\it loc-gowalla} & 196K & 950K & 2,267 & 53.18 & 42.6 \\
{\it web-NotreDame} & 325K & 1,090K & 2,845 & 53.06 & 53.6\\
{\it amazon0601} & 403K & 2,443K & 14,903 & 298 & 50.0 \\
{\it web-Google} & 875K & 4,322K & 65,306 & 824 & 79.2\\
{\it wiki-Talk} & 2,394K & 4,659K & 175,450 & 922 & 190.1 \\
{\it DBLP-coauthor} & 1,236K & 9,081K & 115,919 & 251 & 460.8\\\hline
 \multicolumn{5}{|r|}{\bf geometric mean}&{\bf 99.8}\\\hline
 \end{tabular}
}
\vspace*{-1ex}
\caption{\small The graphs used in the experiments. Column {\em Org.} shows
   the initial closeness computation time of $\ccalg$ and {\em Best}
   is the best update time we obtain in case of streaming
   data.}
 \label{tab:graph_prop}
\vspace*{-1ex}
\end{table}


\subsection{Handling topology modifications}

To assess the effectiveness of our algorithms, we need to know that when each
edge is inserted to/deleted from the graph. Our datasets from UFL Sparse Matrix
Collection do not have this information. To conduct our experiments on
these datasets, we delete 1,000 edges from a graph chosen randomly in the following
way: A vertex $u \in V$ is selected randomly~(uniformly), and a vertex $v \in \Gamma_G(u)$
is selected randomly~(uniformly). Since we do not want to change the
connectivity in the graph~(having disconnected components can make
our algorithms much faster and it will not be fair to \ccalg), we discard $uv$
if it is a bridge. If this is not the case we delete it from $G$
and continue. We construct the initial graph by deleting these 1,000
edges. Each edge is then inserted one by one, and
our algorithms are used to recompute the closeness centrality after each
insertion. Beside these random insertion experiments, we also evaluated
our algorithms on a real temporal dataset of the DBLP coauthor
graph\footnote{\url{http://www.informatik.uni-trier.de/~ley/db/}}. In
this graph, there is an edge between two authors if
they published a paper. Publication dates are used as timestamps
of edges. We first constructed the graph for the papers published before
January 1, 2013. Then, we inserted the coauthorship edges of the papers since then.
Although our experiments perform edge insertion, edge
deletion is a very similar process which should give comparable
results.

In addition to \ccalg, we configure our algorithms in four different ways:
\ccalgb only uses biconnected component decomposition (BCD), \ccalgbl uses
BCD and filtering with levels, \ccalgbli uses all three work
filtering techniques including identical vertices. And
\ccalgblih uses all the techniques described in this paper including
SSSP hybridization. 

Table~\ref{tab:results} presents the results of the experiments.The
second column, \ccalg, shows the time to run the full Brandes
algorithm for computing closeness centrality on the original version
of the graph. Columns $3$--$6$ of the table present absolute
runtimes~(in seconds) of the centrality computation algorithms.  The
next four columns, $7$--$10$, give the speedups achieved by each
configuration. For instance, on the average, updating the closeness
values by using \ccalgb on {\em PGPgiantcompo} is $11.5$ times faster
than running \ccalg. Finally the last column gives the overhead of our
algorithms per edge insertion, i.e., the time necessary to detect the
vertices to be updated, and maintain BCD and identical-vertex
classes. Geometric means of these times and speedups are also given to
provide comparison across instances.

\begin{table*}
\center
\smaller
\scalebox{0.90}{
 \begin{tabular}{|l|r|rrrr|rrrr|r|}
 \hline
 &\multicolumn{5}{|c|}{\bf Time~(secs)}&\multicolumn{4}{|c|}{\bf
 Speedups}& \multicolumn{1}{|c|}{\bf Filter}\\
{\bf Graph} & \ccalg & \ccalgb & \ccalgbl & \ccalgbli & \ccalgblih &
 \ccalgb & \ccalgbl & \ccalgbli & \ccalgblih & {\bf time~(secs)}\\\hline
{\it hep-th} & 1.413	&	0.317	&	0.057	&	0.053 &	0.048	&	4.5	&	24.8	&	26.6	& 29.4	& 0.001\\
{\it PGPgiantcompo}& 4.960	&	0.431	&	0.059	& 0.055	&	0.045	&	11.5	&	84.1	&	89.9 &	111.2	& 0.001\\
{\it astro-ph} & 14.567	&	9.431	&	0.809	&	0.645 &	0.359	&	1.5	&	18.0	&	22.6	&  40.5	& 0.004\\
{\it cond-mat-2005}& 77.903	&	39.049	&	5.618	& 4.687	&	2.865	&	2.0	&	13.9	&	16.6 &	27.2	& 0.010\\\hline
Geometric mean & 9.444 & 2.663 & 0.352 & 0.306 & 0.217 & 3.5 & 26.8 & 30.7 & 43.5 & 0.003\\\hline
{\it soc-sign-epinions} & 778.870	&	257.410		&	20.603 	&	19.935	&	6.254	&	3.0	&	37.8	&  	39.1		& 124.5	& 0.041\\
{\it loc-gowalla}& 2,267.187		&	1,270.820		&	132.955	&	135.015	&	53.182	&	1.8	&	17.1	&	16.8		& 42.6	&0.063\\
{\it web-NotreDame} & 2,845.367	&	579.821		&	118.861	& 	83.817	&	53.059	&	4.9	&	23.9	&	33.9 		& 53.6	& 0.050\\
{\it amazon0601} & 14,903.080		&	11,953.680	&      	540.092	&	551.867	&	298.095	&	1.2	&      	27.6	&	27.0		& 50.0	& 0.158\\
{\it web-Google} & 65,306.600		&	22,034.460	&      	2,457.660	&	1,701.249	&	824.417	&	3.0	& 	26.6	&	38.4		& 79.2	& 0.267\\
{\it wiki-Talk} & 175,450.720		&	25,701.710	&   	2,513.041 &  	2,123.096	& 	922.828 	&	6.8	&	69.8	& 	82.6		& 190.1	& 0.491\\
{\it DBLP-coauthor} &  115,919.518  & 	18,501.147	&  	288.269	& 	251.557	& 	252.647 	&	6.2	& 	402.1& 	460.8	& 458.8	& 0.530\\\hline
Geometric mean & 13,884.152		&	4,218.031		& 	315.777 	&	273.036 	&	139.170 	&	3.2	& 	43.9	 & 	50.8		& 99.7	& 0.146\\\hline

\end{tabular}
}
\vspace*{-1ex}
\caption{\small Execution times in seconds of all the algorithms and speedups
when compared with the basic closeness centrality algorithm \ccalg. In
the table \ccalgb is the variant which uses only BCDs, \ccalgbl uses
BCDs and filtering with levels, \ccalgbli uses all three work
filtering techniques including identical vertices. And
\ccalgblih uses all the techniques described in this paper including
SSSP hybridization. }
\label{tab:results}
\vspace*{-1ex}
\end{table*}

The times to compute closeness centrality using \ccalg on the small
graphs range between $1$ to $77$ seconds. On large
graphs, the times range from $13$ minutes to $49$ hours. Clearly,
\ccalg is not suitable for real-time network analysis and management
based on shortest paths and closeness
centrality. When all the techniques are used~(\ccalgblih), the time
necessary to update the closeness centrality values of the small
graphs drops below $3$ seconds per edge insertion. The improvements
range from a factor of $27.2$~({\em cond-mat-2005}) to $111.2$~({\em PGPgiantcompo}),
with an average improvement of $43.5$ across small
instances. On large graphs, the update time per insertion drops below
$16$ minutes for all graphs. The improvements
range from a factor of $42.6$~({\em loc-gowalla}) to $458.8$~({\em DBLP-coauthor}),
with an average of $99.7$.  For all
graphs, the time spent filtering the work is below one second which
indicates that the majority of the time is spent for SSSPs. Note that
this part is pleasingly parallel since each SSSP is independent from
each other.

The overall improvement obtained by the proposed
algorithms is very significant. The speedup obtained by using
BCDs~(\ccalgb) are $3.5$ and $3.2$ on the average for small and large
graphs, respectively. The graphs {\em PGPgiantcompo}, and
{\em wiki-Talk} benefits the most from BCDs~(with speedups
$11.5$ and $6.8$, respectively). Clearly using the biconnected
component decomposition improves the update performance. However,
filtering by level differences is the most efficient technique: \ccalgbl brings
major improvements over \ccalgb. For all social
networks, \ccalgbl increased the performance when compared with \ccalgb,
the speedups range from $4.8$~({\em web-NotreDame}) to
$64$~({\em DBLP-coauthor}). Overall, \ccalgbl brings a $7.61$
improvement on small graphs and a $13.44$ improvement on large graphs
over \ccalg. 

For each added edge $uv$, let $X$ be the random variable equal to
$|\dis_G(u,w) - \dis_G(v,w)|$. By using 1,000 $uv$ edges, we computed
the probabilities of the three cases we investigated before and give
them in Fig.~\ref{fig:leveldiffs01}. For each graph in the figure, the
sum of first two columns gives the ratio of the vertices not updated
by \ccalgbl. For the networks in the figure, not even $20\%$ of
the vertices require an update~($Pr(X > 1)$). This explains the
speedup achieved by filtering using level differences.
Therefore, level filtering is more useful for the graphs having 
characteristics similar to small-world networks.

\begin{figure}[htbp]
\vspace*{-1ex}
\center
\includegraphics[width=0.44\textwidth]{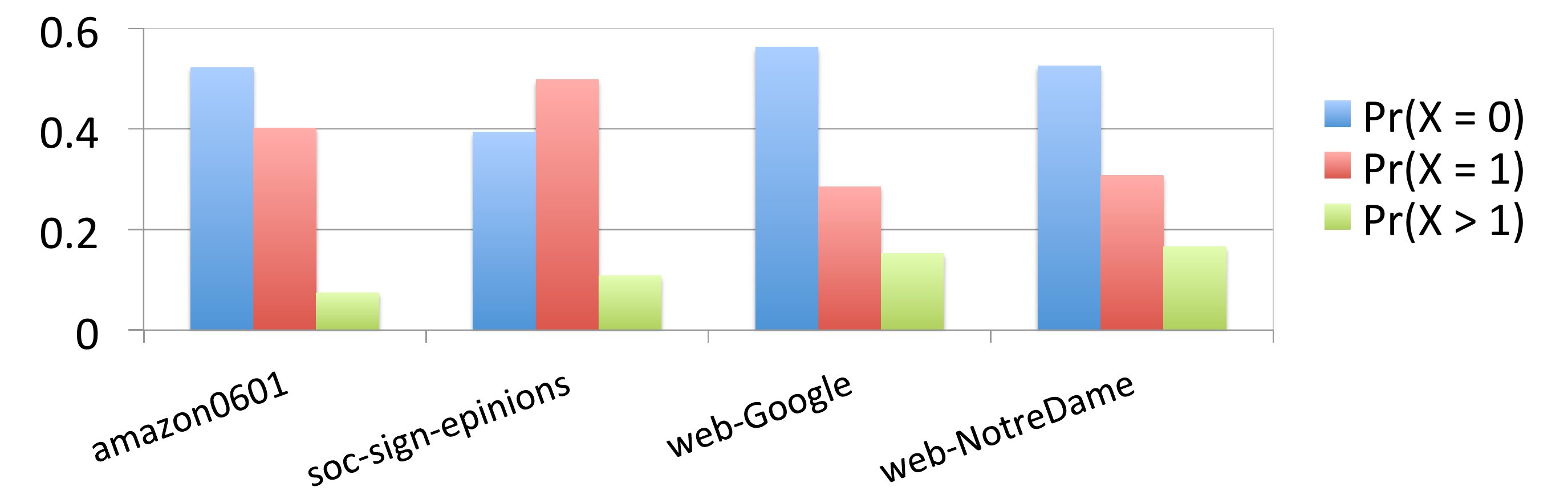}
\vspace*{-1ex}
\caption{\small The bars show the distribution of random variable $X =
  |\dis_G(u,w) - \dis_G(v,w)|$ into three cases we investigated when
  an edge $uv$ is added. 
}
\label{fig:leveldiffs01}
\vspace*{-1ex}
\end{figure}

Filtering with identical vertices is not as useful as the other two
techniques in the work filter. Overall, there is a $1.15$ times
improvement with \ccalgbli on both small and large graphs compared
to \ccalgbl. For some graphs, such as {\em web-NotreDame}
and {\em web-Google}, improvements are much
higher~($30\%$ and $31\%$, respectively).  

Finally, the hybrid implementation of SSSP also proved to be
useful. \ccalgblih is faster than \ccalgbli by a factor of $1.42$ on
small graphs and by a factor of $1.96$ on large graphs. Although it
seems to improve the performance for all graphs, in some few cases,
the performance is not improved significantly.
This can be attributed
to incorrect decisions on SSSP variant
to be used. Indeed, we did not
benchmark the architecture
to discover the proper
parameter. \ccalgblih performs the best on social network graphs with
an improvement ratio of $3.18$~({\em soc-sign-epinions}), $2.54$~({\em
loc-gowalla}), and $2.30$~({\em wiki-Talk}).

All the previous results present the average update time for 1,000
successively added edges. Hence, they do not say anything about the
variance. Figure~\ref{fig:updatedist} shows the runtimes of \ccalgb
and \ccalgblih per edge insertion for {\em web-NotreDame} in a sorted
order. The runtime distribution of \ccalgb clearly has multiple
modes. Either the runtime is lower than $100$ milliseconds or it is
around $700$ seconds. We see here the benefit of BCD. According to
the runtime distribution, about $59\%$ of {\em web-NotreDame}'s
vertices are inside small biconnected components. Hence, the time per
edge insertion drops from 2,845 seconds to 700. 
Indeed, the largest component only contains $41\%$ of the
vertices and $76\%$ of the edges of the original graph. The decrease
in the size of the components accounts for the gain of performance.

\begin{figure}[htbp]
\vspace*{-1ex}
\center
\includegraphics[width=0.44\textwidth]{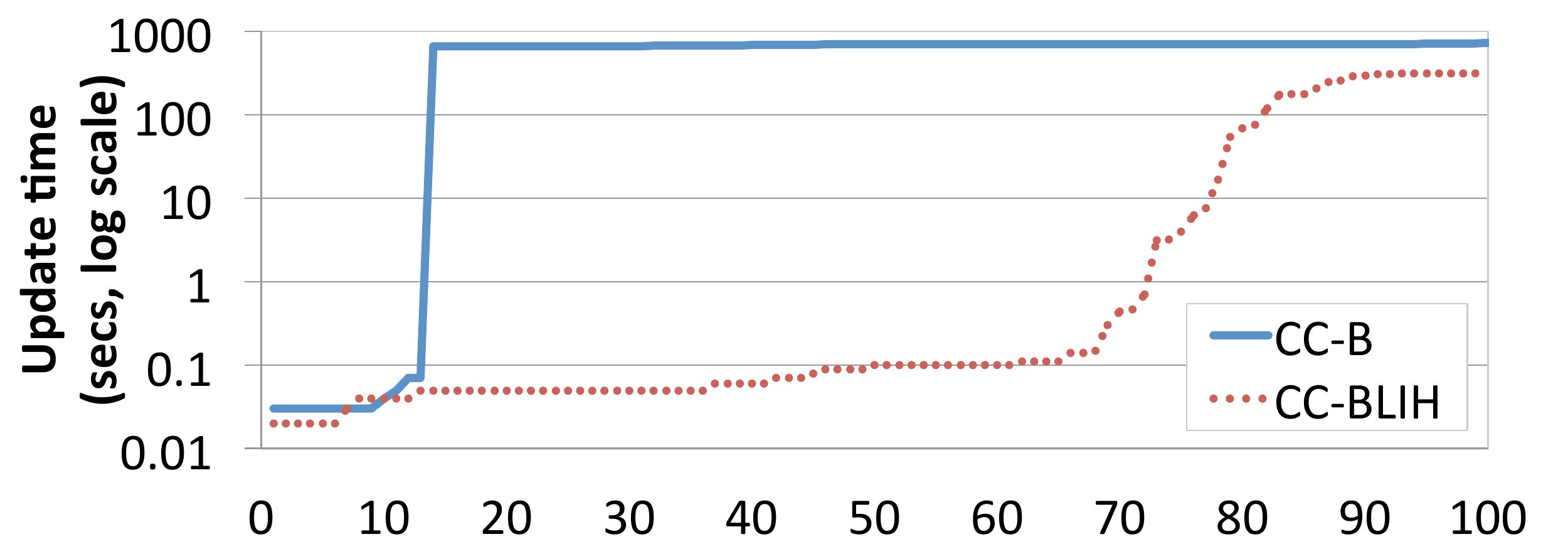}
\vspace*{-1ex}
\caption{\small Sorted list of the runtimes per edge insertion for the first
  $100$ added edges of {\em web-NotreDame}.}
\label{fig:updatedist}
\vspace*{-1ex}
\end{figure}

The impact of level filtering can also be seen on
Figure~\ref{fig:updatedist}. $60\%$ of the edges in the main
biconnected component do not change the closeness values of many
vertices and the updates that are induced by their addition take less
than $1$ second. The remaining edges trigger more expensive updates
upon insertion. Within these $30\%$ expensive edge insertions,
identical vertices and SSSP hybridization provide a significant
improvement~(not shown in the figure).

\paragraph{Better Speedups on Real Temporal Data}

The best speedups are obtained on the DBLP coauthor network, which
uses real temporal data.  Using \ccalgb, we reach $6.2$ speedup w.r.t.
\ccalg, which is bigger than the average speedup on all networks. Main
reason for this behavior is that $10\%$ of the inserted edges are
actually the new vertices joining to the network, i.e., authors with
their first publication, and \ccalgb handles these edges quite
fast. Applying \ccalgbl gives a $64.8$ speedup over \ccalgb, which is
drastically higher than on all other graphs.  Indeed, only $0.7\%$ of
the vertices require to run a SSSP algorithm when an edge is inserted
on the DBLP network. For the synthetic cases, this number is $12\%$.
\ccalgbli provides similar speedups with random insertions and
\ccalgblih does not provide speedups because of the structure of the graph. Overall,
speedups obtained with real temporal data reaches $460.8$, i.e., $4.6$ times
greater than the average speedup on all graphs. Our algorithms appears
to perform much better on real applications than on synthetic ones.   

\subsection{Summary}

All the techniques presented in this paper allow to update closeness
centrality faster than the non-incremental algorithm presented
in~\cite{brandes2001} by a factor of $43.5$ on small graphs and $99.7$
on large ones. Small-world networks such as social networks benefit
very well from the proposed techniques. They tend to have a
biconnected component structure that allow to gain some improvement
using \ccalgb. However, they usually have a large biconnected
component and still, most of the gain is derived from exploiting their
spike-shaped distance distribution which brings at least a factor of
$13.4$. Identical vertices typically brings a small amount of
improvement but helps to increase the performance during expensive
updates. Using all the techniques, we achieved to reduce the closeness
centrality update time from $2$ days to $16$ minutes for the graph
with the most vertices in our dataset~({\em wiki-Talk}). And for the
temporal DBLP coauthorship graph, which has the most edges, we reduced
the centrality update time from 1.3 days to 4.2 minutes.

\section{Conclusion}\label{sec:con}


In this paper we propose the first algorithms to achieve fast updates
of exact centrality values on incremental network modification at such
a large scale.  Our techniques exploit the biconnected component decomposition
of these networks, their spike-shaped shortest-distance distributions,
and the existence of nodes with identical neighborhood. In large
networks with more than $500K$ edges, our techniques proved to bring a
$99$ times speedup in average. 
With a speedup of 458, the proposed techniques may even allow DBLP to
reflect the impact on centrality of the papers published in quasi
real-time. Our algorithms will serve as a fundamental building block
for the centrality-based network management problem, closeness
centrality computations on dynamic/streaming networks, and their
temporal analysis.

The techniques presented in this paper can directly be extended in two
ways. First, using a statistical sampling to compute an approximation
of closeness centrality only requires a minor adaptation on the SSSP
kernel to compute the contribution of the source vertex to other
vertices instead of its own centrality. Second, the techniques
presented here also apply to betweenness centrality with minor
adaptations.

As a future work, we plan to investigate
local search techniques for the centrality-based network management problem
using our incremental centrality computation algorithms.


\section{Acknowledgments}
This work was supported in parts by the DOE grant DE-FC02-06ER2775 and by
the NSF grants CNS-0643969, OCI-0904809, and OCI-0904802.

\bibliographystyle{abbrv}

\end{document}